\documentclass[10pt,english,journal,a4paper]{IEEEtran}
\usepackage{comment}
\usepackage{psfrag}
\usepackage{amsfonts}
\usepackage{amsmath}
\usepackage{amssymb}
\usepackage{amsthm}
\usepackage{mathtools}
\usepackage{alltt}
\usepackage{hyperref}
\usepackage{longtable}
\usepackage{subcaption}
\usepackage{epsfig}
\usepackage{epstopdf}
\usepackage{float}
\usepackage{fancyvrb}
\usepackage[verbatim]{}
\usepackage{multirow}
\usepackage{booktabs}
\usepackage{tabularx}
\usepackage{subcaption}
\usepackage{multicol}
\usepackage{listings}
\usepackage{wrapfig}
\usepackage[noadjust]{cite}
\usepackage{balance}

\newtheoremstyle{mytheoremstyle}{2pt}{1pt}{\itshape}{}{\bfseries}{.}{.5em}{} 
\theoremstyle{mytheoremstyle}
\usepackage{xpatch}
\makeatletter
\xpatchcmd{\proof}{\topsep6\p@\@plus6\p@\relax}{}{}{}
\makeatother

\usepackage[T1]{fontenc}
\makeatletter
\DeclareFontFamily{T1}{lcmtt}{\hyphenchar\font\m@ne}
\makeatother

\DeclareFontShape{T1}{lcmtt}{m}{n}{%
	<13.82><16.59><9><23.89><28.66><34.4><41.28>%
	ecltt8}{}

\usepackage[ruled,vlined,lined,linesnumbered]{algorithm2e} 
\SetKwRepeat{Do}{do}{while}
\usepackage{algpseudocode}
\usepackage{enumerate}
\usepackage{mathtools}
\usepackage{enumitem}

\setlist[description]{leftmargin=\parindent,labelindent=\parindent}
\usepackage[table]{xcolor}
\usepackage{varwidth}
\usepackage{epstopdf}
\newtheorem{definition}{Definition}

\newtheorem{theorem}{Theorem}
\newtheorem{example}{Example}

\newtheorem{property}{Property}

\newcommand{\subparagraph}{}

\newcommand*{\QEDA}{\ensuremath{\qed}}%
\usepackage{titlesec}

{\par\kern\topsep}

\titlespacing*{\section}
{0pt}{1ex plus 0ex minus 1ex}{1ex plus 0ex minus 0.5ex}
\titlespacing*{\subsection}
{0pt}{0.5ex plus 0ex minus 0.5ex}{0.5ex plus 0ex minus 0.8ex}
\titlespacing*{\subsubsection}{0pt}{0.5ex plus 0pt minus 0.5ex}{0.5ex plus 0pt minus 0.5ex}
\usepackage{pdfpages}
\title{Recurrence in Dense-time AMS Assertions}

\author{Sayandeep Sanyal, \IEEEmembership{Student Member, IEEE}, Antonio~Anastasio~Bruto~da~Costa, \IEEEmembership{Student Member, IEEE}, \\
	Pallab Dasgupta, \IEEEmembership{Senior Member, IEEE} \\
	Formal Methods Laboratory, Department of Computer Science and Engineering,\\ 
	Indian Institute of Technology Kharagpur
	\thanks{The authors thank Intel corporation CAD SRS funding for partial support of this research. The authors also acknowledge Antara Ain for building the initial framework for the tool.}
}

\usepackage{etoolbox}
\makeatletter
\patchcmd{\@maketitle}
{\addvspace{0.5\baselineskip}\egroup}
{\addvspace{-1\baselineskip}\egroup}
{}
{}
\makeatother

\begin{document}
	
	\maketitle
	
	\begin{abstract}
	The notion of recurrence over continuous or dense time, as required for expressing
	{\em Analog and Mixed-Signal} (AMS) behaviours, is fundamentally different from what is
	offered by the recurrence operators of SystemVerilog Assertions (SVA). This article 
	introduces the formal semantics of recurrence over dense time and provides a 
	methodology for the runtime verification of such properties using interval arithmetic. 
	Our property language extends SVA with dense real-time intervals and predicates 
	containing real-valued signals. We provide a tool kit which interfaces with 
	off-the-shelf EDA tools through standard VPI. 
	\end{abstract}
	
	
	\IEEEpeerreviewmaketitle
	
	\section{Introduction}\label{sec:introduction}
	
	Assertion (property) language standards, such as SVA~\cite{sva} and PSL~\cite{psl}, are widely 
	used in digital design verification, and some of their recent features enable the 
	specification of properties using real variables, including analog signals sampled at 
	discrete clock boundaries. Properties in SVA are based on \textit{sequence expressions} 
	that capture sequences of Boolean events separated by clock-cycle delays. 
	
	The clocked semantics of SVA can lead to precision related problems when dealing with
	analog signals. For example, consider the property: \textit{"The output $V_{out}$ must 
	cross 1.8~V within $2 \mu s$ to $4.25 \mu s$ of the input $V_{in}$ crossing 3V".} If the
	clock, {\tt clk}, has a period of $0.4 \mu s$, we may write this assertion in SVA as 
	follows: 
	\begin{verbatim}
	wire x, y;
	assign x = V(Vin) > 3;
	assign y = V(Vout) > 1.8;
	property DelayCheck;
	  @(posedge clk) $rose(x) |-> ##[5:11] $rose(y);
	endproperty
	assert property (DelayCheck);
	\end{verbatim}
	Note that the real time interval, $[2 \mu s\ :\ 4.25 \mu s]$, is approximated by the
	discrete interval $[5:11]$ in terms of the number of clock cycles of the clock, 
	{\tt clk}, of period $0.4 \mu s$. The loss of precision due to this approximation
	may lead to missing a failure as shown in Fig.~\ref{fig:precisionErrors}. Here
	$V_{in}$ crossed 3V exactly $4.3 \mu s$ after $V_{out}$ crossed 1.8~V, which exceeded
	the real time interval, but the assertion checker detects that the crossing took place
	within the specified number of clock cycles. 
	
	In general over-approximation (equivalently, under-approximation) produces false 
	positives (equivalently, false negatives). Increasing the precision of the assertion 
	clock reduces the chance of a false positive/negative, but it increases the number of 
	cycles in the interval, and thereby the assertion checking overhead. For this reason,
	existing literature, including our own~\cite{Mukhopadhyay:2009, AMS_LTL, AMS_LTL-L, assertion-aware-sampling, debugging-window, localVariablesAntara, dyfet, forfet}, advocates the use of 
	dense time assertion checking, which works using real interval arithmetic as opposed
	to cycle based reasoning.
    
	\begin{figure}[h]
		\centering
		\includegraphics[width=\linewidth]{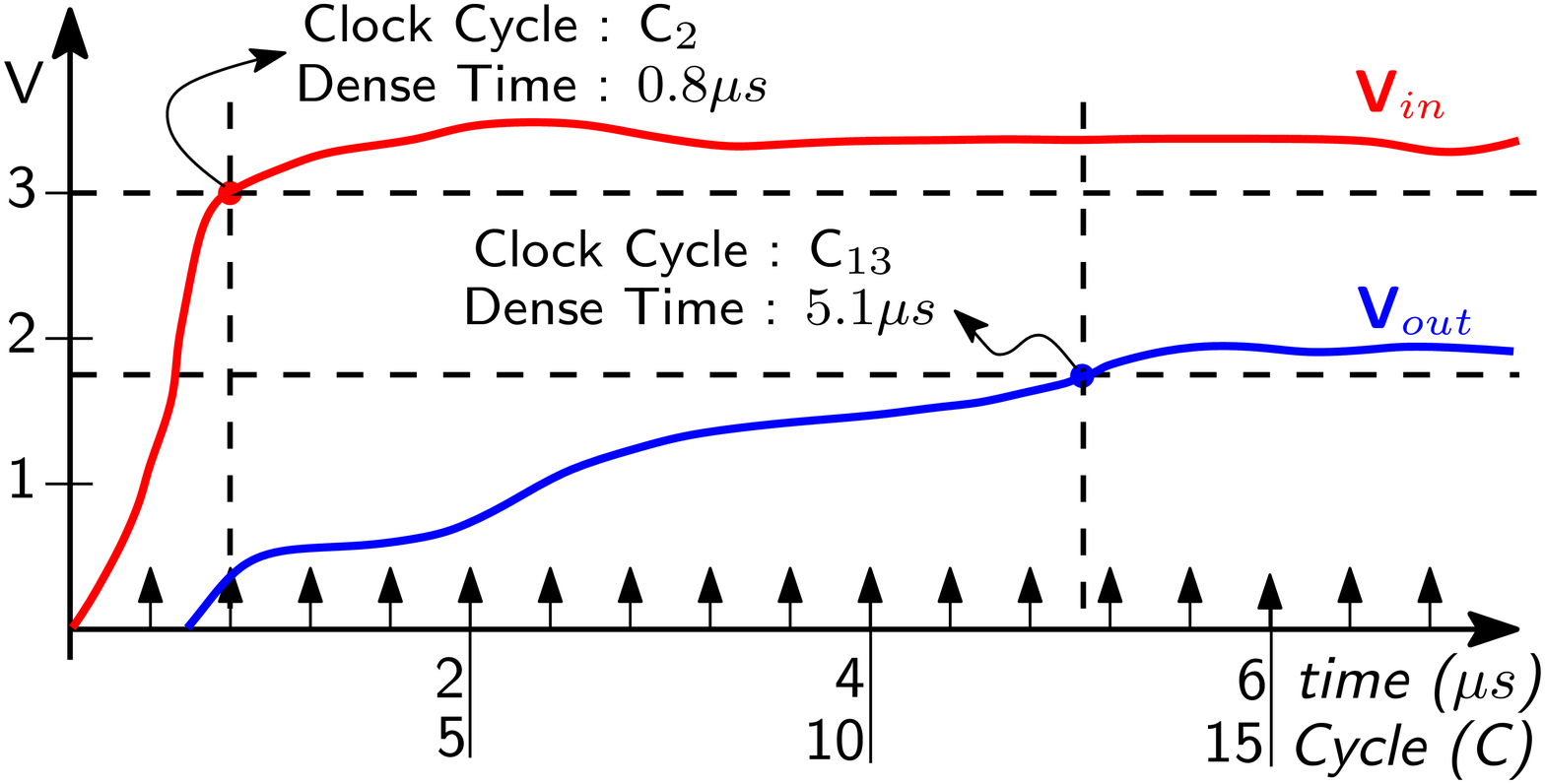}
		\caption{Bugs can escape due to low precision sampling}\label{fig:precisionErrors}
	\end{figure}
	
	The semantics of dense-time naturally allows properties to hold continuously over a
	dense time period. We define \textit{recurrence} to mean that the truth of an expression 
	holds true continuously over a period of time. For example, consider the
	requirement, {\em "If the enable becomes true, and thereafter within 5ms the output 
	voltage is above 3V for at least 2ms, then within the following 0.7ms the out\_good 
	signal stays true for at least 1ms."} An intrinsic feature of such requirements is 
	that predicates must be true continuously over a time period. This is fundamentally 
	different from the notion of recurrence in languages like SVA~\cite{sva} and 
	PSL~\cite{psl}, where recurrence means a countable non-overlapping series of
	matches of a sequence expression.
	
	In this article, we propose the dense time semantics of recurrence and our methodology
	for evaluating assertions having recurrence operators as well as other operators. Since
	recurrence operators are frequently used with other types of operators in an assertion,
	we provide the integrated set of interval arithmetic steps used in our tool, CHAMS.
	Results on CHAMS are also provided at the end.

\section{Recurrence in AMS Assertions}\label{sec:motivatingExample}

This section demonstrates the use of recurrence operators over dense-time. Consider the
waveforms for the voltages of analog nets, $a$, $b$, and $c$, shown in 
Fig.~\ref{fig:motivatingExamples}. Also shown are truth intervals of some of the
{\em predicates over real variables} (PORVs), and the match/fail intervals of the
assertions described below.

\begin{figure}[t]
	\centering
	\includegraphics[width=\linewidth]{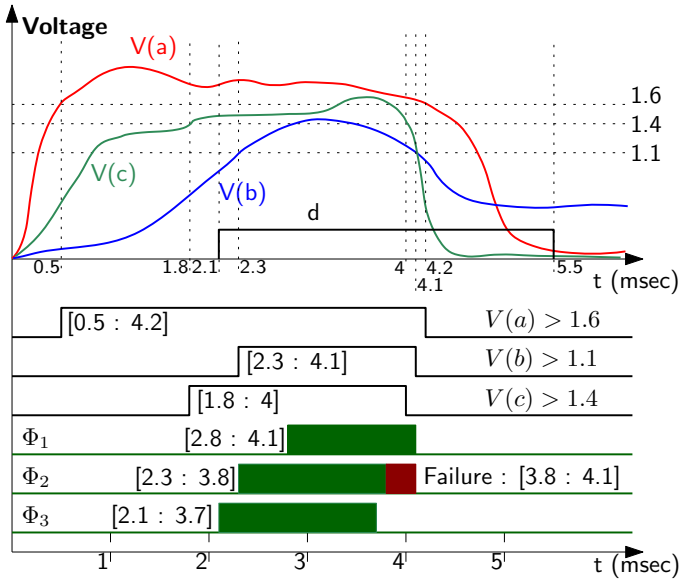}
	\caption{Waveforms of signals and assertion matches}\label{fig:motivatingExamples}
\end{figure}
	
\begin{example}
If $V(a)$ remains above 1.6~V for 2.3~ms, then $V(b)$ will be above 1.1~V.
\begin{lstlisting}[basicstyle=\ttfamily,columns=fullflexible,gobble=2,escapechar=`]
    `$\Phi_1:$ `{V(a)>1.6}[*0.0023] |-> {V(b)>1.1};
\end{lstlisting}
\end{example}

The syntax of the recurrence operator is similar to that of SVA, but the semantics of
recurrence of the PORV, {\tt V(a)>1.6}, is continuous over the specified dense time. The unit
of voltage is volts and the unit of time is seconds.

\begin{example}
Whenever V(b) is greater than 1.1~V, d will be high at some later time
    between 0.2~msec and 1.3~msec, and V(c) will remain above 1.4V until d 
    becomes true.
    \begin{lstlisting}[basicstyle=\ttfamily,columns=fullflexible,gobble=2,escapechar=`]
    `$\Phi_2:$` {V(b)>1.1}|-> 
      {V(c)>1.4}[*0.0002:0.0013]{d==1};
     \end{lstlisting}
    \end{example}
    
    \begin{example}
    If V(a) is higher than 1.6~V and within 1.3~msec to 2.9~msec V(c) is higher than
    1.4~V for 0.3~msec, then V(b) will be greater than $1.1$ at some time later between 
    0.4~ms and 1.2~ms and d will be high until {\tt V(b)>1.1}.
    \begin{lstlisting}[basicstyle=\ttfamily,columns=fullflexible,gobble=2,escapechar=`]
    `$\Phi_3:$ `{V(a)>1.6}##[0.0013:0.0029]{V(c)>1.4}[*0.0003] 
         |-> {d==1} [*0.0004:0.0012] {V(b)>1.1};
    \end{lstlisting}
    \end{example}

The non-vacuous matches of these assertions are shown in Fig.~\ref{fig:motivatingExamples}. 
The assertion $\varphi_2$ also has failures. It may be noted that AMS assertions can match
or fail continuously over a period of time.
	
\section{Formal Semantics}\label{sec:AMS-SVA}

As in SVA, our language, Analog Mixed-Signal Assertion Language (AMSAL), uses the notion of {\em sequence expressions}. A property in 
AMSAL takes one of the following forms:\\
    {\tt SEQ |-> SEQ}\\
    {\tt SEQ |-> \#\#[a:b] SEQ }\\
where:\\
\setlength{\tabcolsep}{2pt}
\begin{tabular}{lcl}
		{\tt SEQ} & {\tt $\rightarrow$} & {\tt BEXPR} \\
		& & {\tt | $ \thicksim  $ \{BEXPR\}}\\
		& & {\tt | \{SEQ\}[$ ^*a $]} \\
		& & {\tt | SEQ \#\#[$a:b$] SEQ} \\
		& & {\tt | \{SEQ\} [$ ^*a:b $] SEQ} \\	
	\end{tabular}	

\noindent In the syntax, {\tt BEXPR} is a Boolean expression over events and PORVs, and $ a,b \in \mathbb{R}^+, ~b \geq a$. We use symbol $\phi$ for Boolean expressions, $\varphi$ for sequence expressions, and $\Phi$ for properties.

	\begin{definition}
		\textbf{Predicate over Real Variables (PORVs)} \textit{If X=$\{x_1,\ x_2,\ ...,\ x_n\}$ 
			denotes the set of continuous variables, then a Predicate over Real Variables, $P$, 
			may be defined as, $P ::= f(x_1,\ x_2,\ ...,\ x_n)\ \sim 0$, where f is a mapping, 
			$f: \mathbb{R}^n \rightarrow \mathbb{R}$, and $\sim$ is a relational operator such that 
			$\sim\ \in \{>,\geq\}$.}
		~\hfill
		$\QEDA$
	\end{definition}
	
	Other relational operators may be derived using $\sim$ along with appropriate propositional connectives.
	
	\begin{definition}\label{def:event}
		\textbf{Events: }Events are of the form $@^*(P)$, where $P$ is a PORV and $* \in \{+,-,~\}$. $@^+(P)$, 
		$@^-(P)$, $@(P)$ are true respectively at the positive edge, negative edge, both positive
		and negative edge of the truth of $P$.~\hfill$\QEDA$
	\end{definition}
	
	Boolean expressions in AMSAL are written over Boolean propositions, PORVs, and events. The syntax for {\tt BEXPR} is as follows:
	
	{\small \setlength\tabcolsep{3pt}
		\begin{tabular}{rcl}
			{\tt {BEXPR}} & {\tt $\rightarrow$} & {\tt{EVENT}} \\
			& & {\tt | {EVENT, ASGMT}}\\
			& & {\tt | (CONJUNCT)} \\
			& & {\tt | {(CONJUNCT) , ASGMT}} \\
			& & {\tt | {EVENT \&\& (CONJUNCT)}} \\
			& & {\tt | {BEXPR OR BEXPR}}\\
			
			{\tt {CONJUCNT}} & {\tt $\rightarrow$} & {\tt {CONJUNCT  \&\& PORV}}\\
			& & {\tt | {PORV}}\\
			
		\end{tabular}
	}

	\begin{definition}\label{def:satisfaction}
		\textbf{Satisfaction Relation: $\tau(t) \models \varphi$} is defined as follows for Boolean expressions:
		\begin{itemize}
			\item $\tau(t) \models P$ iff $P$ is a PORV, and $P$ is true for signal valuations in state $\tau(t)$ or $P$ is a Boolean signal and $\tau_P(t) = true$.
			\item $\tau(t) \models D$ where $D=P_1 \vee P_2 \vee ... \vee P_N ~$ iff $~\exists_{1\leq i \leq N} ~ \tau(t)  \models P_i$, and $P_i$ is a PORV.
			\item $\tau(t) \models C$ where $C=P_1 \wedge P_2 \wedge ... \wedge P_N ~$ iff $~\forall_{1\leq i \leq N} ~ \tau(t)  \models P_i$, and $P_i$ is a PORV.
			
			\item $ \tau(t) \models E $ where $ E $ is an event on PORV $P$, iff either of the following holds:
			\begin{itemize}
				\item $ E \equiv @^+(P) : \tau(t) \models P ~\bigwedge~ \exists_{t'<t} ~\forall_{\hat{t} \in [t',t)}~ \tau(\hat{t}) \not\models P$
				
				\item $ E \equiv @^-(P) : \tau(t) \not\models P ~\bigwedge~ \exists_{t'<t} ~\forall_{\hat{t} \in [t',t)}~ \tau(\hat{t}) \models P$
				
				\item $ E \equiv @(P) : \tau(t) \models @^+(P) \bigvee \tau(t) \models @^-(P)$
			\end{itemize}
			
		\end{itemize}
		
		~\hfill$\QEDA$
	\end{definition}

We now describe the simulation semantics of recurrence in dense-time as used in AMSAL with 
an interpretation over a simulation trace $\tau$. 
	
	\begin{definition}
		\textbf{Simulation Trace}
		A simulation trace $\tau$ is a mapping 
		$\tau : \mathbb{R}_{\geq 0} \rightarrow \mathbb{R}^{|V|}$, 
		where $V = \{v_1,~v_2,~...,~v_n\}$ is the set of variables 
		(Boolean and Real) representing signals of the system. The state of the system in $\tau$ at time $t$ is given as $\tau(t)$. For $x \in V$, its value at time $t$, in $\tau$, is $\tau_{x}(t)$.	~\hfill$\QEDA$	
	\end{definition}
	
	In Definition~\ref{def:endMatchSat}, state satisfaction for a Boolean expression $\phi$, that is $\tau(t) \models \phi$, is extended to sequence expressions with support for recurrence with dense-time. Note that, a Boolean expression is also a sequence expression.
	
	\begin{definition}
		\label{def:endMatchSat}
		\textbf{Extended Satisfaction Relations:}
		$\tau(t) \models_e \varphi$ is recursively defined for a sequence expression $\varphi$, to be true iff:
		\begin{itemize}
			\item $\{\varphi_1\}${\tt[$^*$a]}, {\tt a$>0:$}
			$(\forall_{t' \in [0:a]} ~ \tau(t-t') ~{\models}_e ~ \varphi_1 ) $

			\item $\varphi_1${\tt \#\#[a:b]}$\varphi_2:$
			$\tau(t)\models_e\varphi_2 \wedge t \in [t'+a:t'+b] \wedge \tau(t') {\models}_e \varphi_1$

			\item  $ \{\varphi_1\}${\tt[$^*$a:b]}$\varphi_2:$
			$\tau(t) \models_e \varphi_2 \wedge \forall_{t' \in [0:a]}  \tau(t-t') {\models}_e  \varphi_1 $

		\end{itemize}
		where $\varphi_1$ and $\varphi_2$ are sequence expressions.
		$~\hfill\QEDA$
	\end{definition}

	The satisfaction of a temporal expression has a begin time and an end time for a match of the expression, known respectively as the \textit{begin match} and \textit{end match} for a sequence expression. The relation $\models_e$ describes the end-match for a sequence expression. Due to the nature of the truth of 	a temporal property over analog artifacts, an end-match at time $t$ for sequence expression $\varphi$ may be associated with multiple begin-match time points, and vice versa. 
	
	\begin{definition}\label{def:beginMatchSat}
		\textbf{Begin Match:} We define $\mathbb{B}(\tau,\varphi,t)$ as the set of time points $\hat{t}$ such that there exists a match of $\varphi$ in $\tau$ starting at $\hat{t}$ and ending at $t$. Formally, for $\tau(t) \models_e \varphi$, $\hat{t} \in \mathbb{B}(\tau,\varphi,t)$ is said to be a begin match for $\varphi$'s end-match at time $t$ iff:
		
		\begin{itemize}
			\item $ \{\varphi_1\}${\tt[$^*$a]}$:$ {\tt a$>0$}, $ t' = t-a, ~\hat{t} \in \mathbb{B}(\tau,\varphi_1,t')$

			\item $\varphi_1${\tt \#\#[a:b]}$\varphi_2$ or $ \{\varphi_1\}${\tt[$^*$a:b]}$\varphi_2:$ $(t'\in\mathbb{B}(\tau,\varphi_2,t)) \wedge({t'' \in [t'-b:t'-a]}) \wedge (\tau(t'') \models_{e} \varphi_1) \wedge ( \hat{t} \in \mathbb{B}(\tau,\varphi_1,t'') ) $

		\end{itemize}
		where $\varphi_1$ and $\varphi_2$ are sequence expressions.
			~\hfill$\QEDA$
	\end{definition}

	\begin{definition}
	\label{def:assertionMatchSemantics}
		\textbf{Match of an assertion:} We say that an assertion has matched at time $t$ iff $t$ is an end match of the antecedent and is a begin match of the consequent. Hence, the assertion $\varphi_1${\tt |->}$\varphi_2$ matches trace $\tau$ 
		non-vacuously at time $t$, iff $\exists_{t' \geq t}~ (\tau(t) \models_e \varphi_1 ~\bigwedge ~ t' \in \mathbb{B}(\tau,\varphi_2,t))$. 
		
		Note that an assertion vacuously match at time $t$, when $\tau(t)\not\models_e \varphi_1$. An assertion fails at time $t$ iff it has no vacuous nor non-vacuous match at time $t$.~\hfill
		$\QEDA$
	\end{definition}
	
	In the following section, we propose an interval abstraction that enables computing the match of assertions over dense time.
	
	\section{Interval Arithmetic for AMSAL}\label{subsec:intervalArith}
	In this section we discuss how the match of a property interpreted over dense-time may be computed, by recursively interpreting the truth of expressions in the property as operations over time intervals.
	
	\begin{definition}
		\textbf{Time Interval:}
		A time interval $ I $ is a nonempty convex subset of $ \mathbb{R}_{\geq 0} $ expressed as $ [a:b] $, $ (a:c) $, $ [a:c) $, and $ (a:c] $ where $ a,b,c \in \mathbb{R}_{\geq 0} $ and $ b \geq a $, $ c>a $. $ l(I) $ and $ r(I) $ are used to denote the left ends and right ends respectively of interval $ I $. ~\hfill$\QEDA$
	\end{definition}
	For an interval $I$, the Minkowski operators are as follows. The Minkowski sum, $ I \oplus [c:d] $, where $c,d \in \mathbb{R}_{\geq 0}$, $ c\leq d$, is computed as, $ I \oplus [c:d] = [l(I)+c:r(I)+d] $. Similarly, the Minkowski difference, $ I \ominus [c:d] $,  where $c,d \in \mathbb{R}_{\geq 0}$, $ c\leq d$, is computed as, $ I \ominus [c:d] = [l(I)-d:r(I)-c]$. Note that any interval $[a:b]$, where $a>b$ is a null interval.
	\begin{definition}
		\textbf{Truth Interval:}
		Time interval $I$ is a truth interval of $\varphi$, where $\varphi $ is a PORV, event, Boolean signal,
		or Boolean expression iff $ \forall_{t \in I} ~ \tau(t) \models \varphi$.
		For each PORV, event, Boolean signal, or Boolean expression, $\varphi$, $\mathcal{I}_\varphi(\tau)$ 
		is the set of all truth intervals of $\varphi$ in $\tau$.
		~\hfill$\QEDA$
	\end{definition}
	
	For trace $\tau$, $I_\tau=[L:R]$ is the time interval over which the trace is defined. In this context, the complement of a truth interval $I$, the \textit{false interval}, is denoted $\bar{I}$ = $\{t'| t'\in I_\tau, t'\not\in I\}$.
	
	In general, a sequence expression may be expressed as:
	{\tt $\phi_1$ $\theta_1$ $\phi_2$ $\theta_2$ $\cdots$ $\theta_{n-1}$ $\phi_n$}, where
	$\forall_{1 \leq i \leq n} ~ \phi_i$ is a Boolean expression of propositions, PORVs and events, and 
	$\forall_{1 \leq j < n} ~ \theta_j$  is a list of sequence operators of the form {\tt \#\#[a:b], [*a:b]} and 
	{\tt [*a]}. If $\theta_j$ contains a sequence of operators, they are always applied to $\phi_j$ from left to 
	right. 
	The following definitions describe the interval arithmetic 
	interpretation for sequence expressions.
	\begin{definition}
		\textbf{Interval Set $\mathbf{I}(\tau)$:}
		For trace $\tau$, $\mathbf{I}(\tau)$ is choice of truth intervals, $\langle  I_{P_1},I_{P_2},...,I_
		{P_k}  \rangle$, $I_{P_i} \in \mathcal{I}_{P_i}(\tau), \forall P_i \in \mathbb{P}_\varphi$, where 
		$\mathbb{P}_\varphi$ is the set of all Boolean propositions, PORVs and events defined in the 
		assertion $\varphi$. For ease of use $\mathbf{I}_P(\tau) = I_P \in \mathbf{I}(\tau)$.
		~\hfill$\QEDA$
	\end{definition}
	
	The truth interval for a sequence expression is viewed as having two contexts, a \textit{begin match} and 
	an \textit{end match} context, as defined by Definitions~\ref{def:endMatchSat} and \ref{def:beginMatchSat}. For non-temporal artifacts such as PORVs, events and Boolean expressions, the two 
	contexts evaluate to the same set of truth intervals. The computation of the begin and end match truth 
	intervals for a sequence expression is computed recursively. As defined earlier, sequence operators are 
	left associative. Brackets may also be used to describe sequences that break away from the default 
	semantics. In general the computation is defined below.
	
	\begin{definition}
		\label{def:beginEndMatch}
		\textbf{Begin and End Match Truth Intervals:}
		Given a choice $\mathbf{I}(\tau)$ of truth intervals, the end match interval, $\mathcal{M}_E(\varphi, \mathbf{I}(\tau))$, and begin match interval, $\mathcal{M}_B(\varphi,\mathbf{I}(\tau))$, for sequence expression $\varphi$ is defined as follows:
		
		{\setlength\tabcolsep{2pt} \small
			\renewcommand{\arraystretch}{1.4}
			$\mathcal{M}_E(\varphi, \mathbf{I}(\tau))$
			
			\begin{tabular}{cp{9cm}}
				
				= & $\mathcal{M}_E (\hat{\varphi}, \mathbf{I}(\tau)) \cap \mathbf{I}_P(\tau)$  $,\varphi \equiv \hat{\varphi} \bigwedge P $ \\
				
				= & $ \mathcal{M}_E( \hat{\varphi}, \mathbf{I}(\tau) ) \cup \mathbf{I}_P(\tau) $
				$, \varphi \equiv \hat{\varphi} \bigvee P $\\
				
				= & $[~l(\mathcal{M}_E( {\varphi_1}, \mathbf{I}(\tau) ) ) + a:r(\mathcal{M}_E( {\varphi_1}, \mathbf{I}(\tau) ) )~]$
				$, \varphi \equiv \{\varphi_1\} ${\tt[$^*a$]}\\
				
				= & $( \mathcal{M}_E(\varphi_1, \mathbf{I}(\tau)) \oplus [a:b] ) \cap \mathcal{M}_E(\varphi_2, \mathbf{I}(\tau))$
				$, \varphi \equiv \varphi_1$~{\tt\#\#[a:b]}$\varphi_2$\\
				
				= & $\mathcal{M}_E(\{\varphi_1\} {\tt [*a]}, \mathbf{I}(\tau) )~ \cap ~\mathcal{M}_E (\varphi_2, \mathbf{I}(\tau))$
				$, \varphi \equiv \{\varphi_1\}$~{\tt[$^*a$:b]}$\varphi_2$\\
				\\
			\end{tabular}
			
			$\mathcal{M}_B(\varphi, \mathbf{I}(\tau))$
			
			\hspace{-7pt}\begin{tabular}{cp{8.4cm}l}
				= & $\mathcal{M}_B (\hat{\varphi}, \mathbf{I}(\tau)) \cap \mathbf{I}_P(\tau)$  $,\varphi \equiv \hat{\varphi} \bigwedge P $ \\
				
				= & $ \mathcal{M}_B( \hat{\varphi}, \mathbf{I}(\tau) ) \cup \mathbf{I}_P(\tau)$  $, \varphi \equiv \hat{\varphi} \bigvee P $\\
				
				= & $\mathcal{M}_B(\varphi_1,[l(\mathcal{M}_E( {\varphi}, \mathbf{I}(\tau) ) ):r(\mathcal{M}_E( {\varphi}, \mathbf{I}(\tau) ) )]-a)$ 
				$, \varphi \equiv {\varphi_1} ${\tt[$^*$a]}\\
				
				= & $\mathcal{M}_B(\varphi_1 , ( \mathcal{M}_E(\varphi, \mathbf{I}(\tau)) \ominus [a:b] ) \cap  \mathcal{M}_E(\varphi_1, \mathbf{I}(\tau))),$\\
				
				& \hfill $\varphi \equiv \varphi_1${\tt\#\#[a:b]}$\varphi_2$ or $\varphi \equiv \{\varphi_1\}${\tt[$^*$a:b]}$\varphi_2$ \\
				
		\end{tabular}
			\renewcommand{\arraystretch}{1}
		}
		~\hfill $\QEDA$
	\end{definition}

    We recursively prove that the arithmetic defined in Definition~\ref{def:beginEndMatch} correctly computes the time points defined by assertion match semantics in Definitions~\ref{def:endMatchSat} and ~\ref{def:beginMatchSat}.
    
    The fundamental case, when $\varphi$ is Boolean ({\tt BEXPR}), is straightforward. The intervals of truth for the $\wedge$ and $\vee$ operations are respectively computed using the $\cap$ and $\cup$ set operations over intervals \cite{temporal properties}. In these cases, the begin and end-matches are identical. 
    For each of the four semantic rules, let $\tau(t)\models_e \varphi$, that is $t$ is an end-match time point. We prove the arithmetic assuming the interval set $\mathbf{I}(\tau)$, a labelled set of truth intervals, one for each $P\in\mathbb{P}_\varphi$.  Let $\mathcal{M}_E(\varphi,\mathbf{I}(\tau))=[l:r]$ be a non-empty end-match interval. We use $\mathcal{D}(.):\mathbb{R^+} \rightarrow I$, to generalize a time point to a time interval, in the context of a quantifier.
    \begin{theorem}
    The end-match intervals computed in Definition~\ref{def:beginEndMatch} correctly compute matches according to the semantics defined in Definition~\ref{def:endMatchSat}.
    \end{theorem}
    \begin{proof}
    Let $\mathcal{M}_E(\varphi_1,\mathbf{I}(\tau))=[l_1:r_1]$ and $\mathcal{M}_E(\varphi_2,\mathbf{I}(\tau))=[l_2:r_2]$ be end-match truth intervals for $\varphi_1$ and $\varphi_2$.
    \begin{itemize}
        \item $\varphi \equiv \{\varphi_1\}${\tt [*a]}: 
        
        $\mathcal{D}(t')=[0:a]$ and $\mathcal{D}(t)=[l,r]$. 
        Since $\tau(t-t')\models_e \varphi_1$, $\mathcal{D}(t-t')=[l_1,r_1]$. Hence, for a non-empty end-match interval for $\varphi$, $r_1-l_1\geq a$. Also, the earliest time point in $\mathcal{M}_E(\varphi,\mathbf{I}(\tau))$ is $l = min_{(t-t') \in \mathcal{D}(\varphi_1)}(t-t') + max_{t'\in[0:a]}(t') = l_1+a$. While the latest time point is
        $r = max_{(t-t') \in \mathcal{D}(\varphi_1)}(t-t') = r_1$.
        
        \item $\varphi \equiv \varphi_1$ {\tt \#\#[a:b] }$\varphi_2$: 
        Using the domains of $t$ and $t'$, we have, $\mathcal{D}(t) = $ $[l_2:r_2] \cap (\mathcal{D}(t')\oplus[a:b])$
        $= [l_2:r_2] \cap [l_1 + a:r_1 + b]$
        
        \item $\varphi \equiv \{\varphi_1\}${\tt[$^*$a:b]}$\varphi_2 $:  
        As shown earlier, for a non-empty end-match, the right-hand side of the conjunction evaluated to the interval $[l_1+a:r_1]$. Hence $\mathcal{D}(t) = [l_2,r_2] \cap [l_1+a:r_1]\hfill\qedhere$.

    \end{itemize}
    \end{proof}

    \begin{theorem}
    The begin-match intervals computed in Definition~\ref{def:beginEndMatch} correctly compute matches according to the semantics defined in Definition~\ref{def:beginMatchSat}.
    \end{theorem}
    \begin{proof}
    Let $\hat{t} \in \mathbb{B}(\tau,\varphi,t)$. We compute $\mathcal{D}(\hat{t})$ as follows:
    \begin{itemize}
        \item $\varphi \equiv \{\varphi_1\}${\tt [*a]}: 
        $\mathcal{D}(t') = \mathcal{D}(t) - a = [l-a : r -a]$. So $\mathcal{D}(\hat{t}) = \mathcal{M}_B(\varphi_1, [l-a : r-a])$
        
        \item $\varphi \equiv \varphi_1$ {\tt \#\#[a:b] }$\varphi_2$ or $\{\varphi_1\}${\tt[$^*$a:b]}$\varphi_2 $: 
        $\mathcal{D}(t) = [l:r] $,
        $\mathcal{D}(t'') = ([l:r] \ominus[a :b]) \cap \mathcal{M}_E(\varphi_1, \mathcal{D}(t)) $,
        hence $\mathcal{D}(\hat{t}) = \mathcal{M}_B(\varphi_1, \mathcal{D}(t'')) $\hfill\qedhere
    \end{itemize}
    \end{proof}
		
		\begin{definition}
		\label{def:assertionMatchInterval}
		\textbf{Match Truth Interval for an Assertion :} Given a choice of $\mathbf{I}(\tau)$ on a simulation trace $\tau$, the match truth interval $I_M$ for an assertion $\Phi$ is computed as follows:
		\begin{itemize}
			\item $ \Phi : \varphi_1$ {\tt|->} $\varphi_2$, $I_M =  \mathcal{M}_E(\varphi_1, \mathbf{I}(\tau))  \cap  \mathcal{M}_B(\varphi_2, \mathbf{I}(\tau)) $.
			
			\item  $ \Phi : \varphi_1${\tt|-> \#\#[a:b]}$\varphi_2 $, $I_M = \mathcal{M}_E(\varphi_1, \mathbf{I}(\tau)) \cap (\mathcal{M}_B(\varphi_2, \mathbf{I}(\tau)) \ominus [a:b]) $.
			\hfill$\QEDA$
		\end{itemize}
	\end{definition}
	
    \begin{theorem}	
    The match truth interval computed in Definition \ref{def:assertionMatchInterval} correctly computes matches according to the semantics defined in Definition~\ref{def:assertionMatchSemantics}.
    \end{theorem}
    \begin{proof}
    Consider each statement in Definition~\ref{def:assertionMatchSemantics}.
    \begin{itemize}
        \item $ \Phi : \varphi_1$ {\tt|->} $\varphi_2$, Follows directly from Definition \ref{def:assertionMatchSemantics}
        
        \item  $ \Phi : \varphi_1${\tt|->\#\#[a:b]}$\varphi_2 $, $t \in I_M$ iff $t \models_e \varphi_1$ and $\exists t', t''$, $ t' \in \mathcal{B}(\tau, \varphi_2, t'') $ for $t''>t'$ and $t \in [t'-b : t'-a]$.
        Hence $I_M = \mathcal{M}_E(\varphi_1, \mathbf{I}(\tau)) \cap (\mathcal{M}_B(\varphi_2, \mathbf{I}(\tau)) \ominus [a:b]) $
        \hfill\qedhere
    \end{itemize}
    
    \end{proof}
	
	\subsection{Evaluating Property Matches}
	Consider the property describing the settling time of the output voltage {\tt Vout} for an arbitrary 
	circuit, as given below. Given the continuum of time when various predicates in the property are true, as 
	shown in Figure~\ref{fig : match}, we describe how the truth of the property is evaluated using the 
	definitions presented in the earlier sections. 
	
	\begin{verbatim}
	property SettlingTime{};
	  @+{V(Vout),0.1*1.2} |-> ##[0.001:0.004] 
	   {V(Vout)>=0.95*1.2 && V(Vout)<=1.05*1.2}[*0.002];
	endproperty
	\end{verbatim}
	
	Let the event in the antecedent be denoted as $E_1$ and the PORV in the consequent is denoted as $ P_1 $. We re-write the assertion as follows:
	{\tt ${\tt E_1}$ |-> \#\#[0.001:0.004] ${\tt P_1}$[$^*$0.002]}, where \\
	{\tt $\tt E_1$ : @${\tt ^+}$\{V(Vout),0.1*1.2\}}, and \\
	{\tt $\tt P_1$ : \{V(Vout)>=0.95*1.2~ \&\& ~V(Vout)<=1.05*1.2\}}
	
	\begin{figure}[t]
		\centering
		\includegraphics[width=\linewidth]{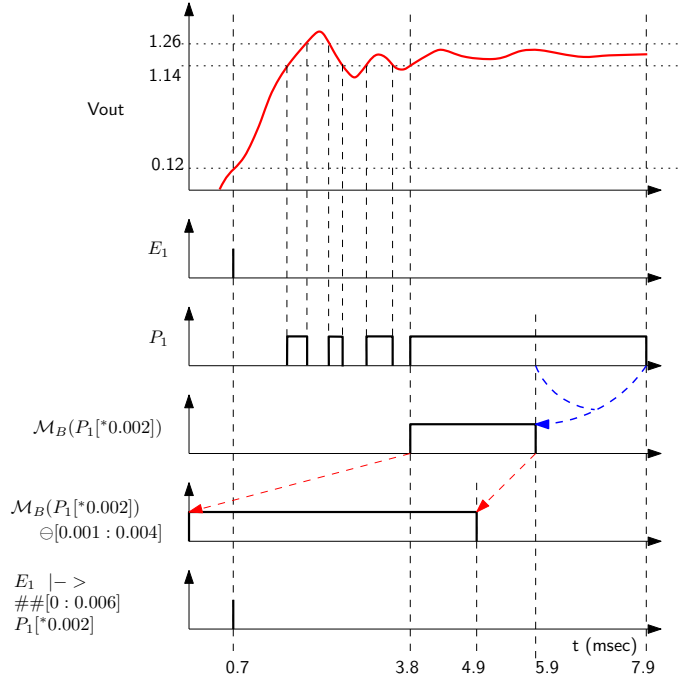}
		\caption{Bottom-up evaluation of assertions.}\label{fig : match}
	\end{figure}
	
	Let for trace $ \tau $, $ \mathcal{I}_{E_1}(\tau) = \langle [0.7ms:0.7ms] \rangle $ and $ \mathcal{I}_{P_1}(\tau) = \langle [1.6ms : 2.03ms], [2.41ms : 2.66ms], [3.04ms : 3.55ms], [3.8ms : 7.9ms] \rangle $. Figure \ref{fig : match} demonstrates the bottom-up 
	approach used for computing the truth of the assertion. The begin and end of $ E_1 $ and $ P_1 $ are 
	computed using Definition~\ref{def:beginEndMatch}. 
	Following the conditions mentioned in Section~\ref{subsec:intervalArith} the assertion is computed to be 
	true for the time interval $ [0.7ms : 0.7ms] $.

	
\section{The CHecker for AMS (CHAMS) Tool}
\label{sec:tool}

In this Section, we describe CHAMS, an online assertion checking tool for AMS. CHAMS works 
with off-the-shelf EDA tools to verify dense time AMS assertions during simulation. This
paper extends CHAMS with recurrence operators. An overview of the tool is shown in 
Fig.~\ref{fig:toolBlockDiag}. 

	\begin{figure}[t]
		\centering
		\includegraphics[width=\linewidth]{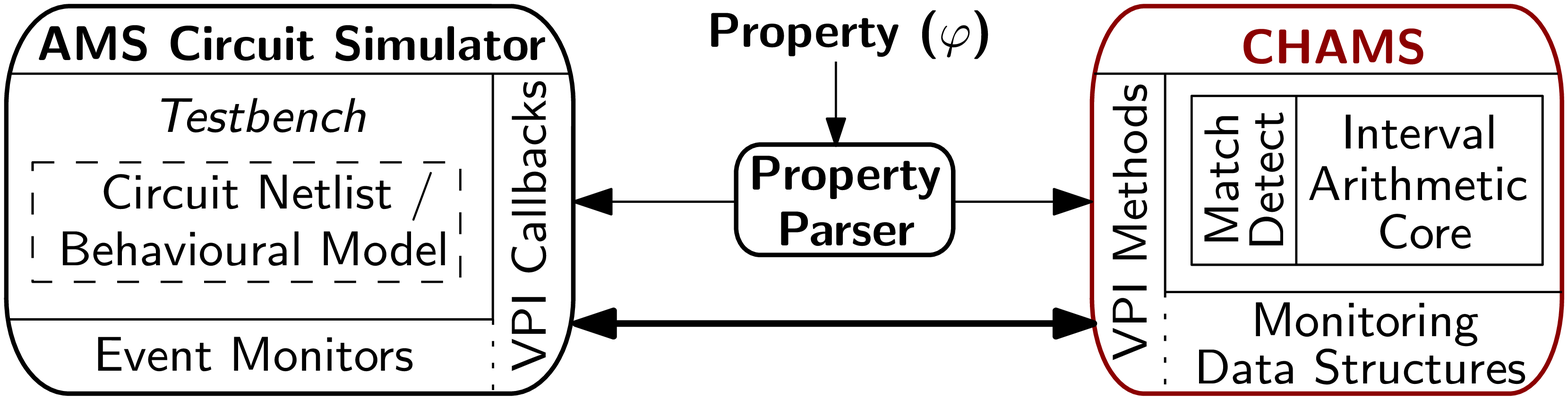}
		\caption{CHAMS Tool Flow}\label{fig:toolBlockDiag}
	\end{figure}	

	\subsection{Inputs}
	\label{subsec:inputs}
	
	The inputs to the tool are a circuit netlist/behavioural model, its testbench and the assertions that 
	need to be checked. Assertions are written in the syntax described in Section \ref{sec:AMS-SVA}. The 
	assertion specification is analysed to automatically generate monitor codes (as explained in the 
	following subsection) containing VPI-callback functions for the monitoring of events and PORVs affecting 
	the truth of the assertion.

	\subsection{Monitor Generation}\label{subsec : monitor_generation}
	
	In order to maintain truth intervals and thereby compute the truth of the assertion, CHAMS generates 
	Verilog-AMS (VAMS) monitors and injects them into the VAMS testbench for the circuit. Note that it is not 
	mandatory to have a VAMS testbench. In its absence monitor codes may also be placed in an independent 
	module having access to circuit ports. 
	
	A monitor consists of standard VPI callbacks which send information about the state of the circuit to the 
	checker CHAMS. CHAMS in turn maintains a data structure in which it updates the truth intervals for each 
	sub-expressions bottom-up. It uses interval arithmetic to do this, and thereby computes the truth of the 
	assertion. We demonstrate this using the property settling time described in the following example.

\begin{example}

\textbf{Rising Sequence:} If the {\tt enable} is asserted and the output voltage {\tt Vout} 
crosses $10\%$ of its rated value of $1.2V$ within $100\mu s$, then thereafter, within 
$1ms$ to $4ms$ {\tt Vout} must reach its {\em steady state} (explained below).
\begin{verbatim}
		property RisingSequence{};
		    @+{enable} ##[0:0.0001] @+{V(Vout),0.1*1.2} 
		        |-> ##[0.001:0.004] 
		    {V(Vout)>=0.95*1.2 && V(Vout)<=1.05*1.2}[*0.002];
		endproperty
\end{verbatim}
{\em 
Here {\em steady state} means, {\tt V(Vout)} remains within $\pm 5\%$ of the rated voltage $1.2V$ for at least $2ms$, and we express this requirement using the recurrence operator in the consequent of the above assertion. In time linear in the length of the property, CHAMS generates VAMS monitors, one for each event, and two for each predicate and Boolean expression. For the property above, the following monitors are generated.}
		\begin{verbatim}
		assign flag_2 = flag_2_0 && flag_2_1;
		always @(posedge enable) 
		    $checkerCall(0,0,$abstime);
		always @(cross(Vout-0.1*1.2,+1,1e-9,1e-6)
		    $checkerCall(0,1,$abstime);
		always @(cross(Vout-0.95*1.2,+1,1e-9,1e-6)
		    flag_2_0 = 1'b1;
		always @(cross(Vout-0.95*1.2,-1,1e-9,1e-6) 
		    flag_2_0 = 1'b0;
		always @(cross(Vout-1.05*1.2,+1,1e-9,1e-6) 
		    flag_2_1 = 1'b0;
		always @(cross(Vout-1.05*1.2,-1,1e-9,1e-6) 
		    flag_2_1 = 1'b1;
		always@(posedge flag_2) 
		    $updateTruthInterval(0,2,+1,$abstime);
		always@(negedge flag_2) 
		    $updateTruthInterval(0,2,-1,$abstime);	\end{verbatim}
{\em 
The callbacks indicated in these monitors invoke the CHAMS assertion checker code 
whenever any event relevant to the assertion is detected during simulation. The first
monitor issues a callback when {\tt enable} goes high, the second monitor issues a
callback when {\tt Vout} crosses $0.1*1.2V$, and the third monitor examines the truth
of the compound predicate {\tt V(Vout)>=0.95*1.2 \&\& V(Vout)<=1.05*1.2}.
} ~~\hfill$\QEDA$
\end{example}

For each sub-expression, CHAMS maintains begin match $(\mathcal{M}_B)$ 
and end match $(\mathcal{M}_E)$ intervals. When the begin and end match for both the antecedent and the consequent are available, the truth of the assertion is evaluated, as given in Definition \ref{def:assertionMatchInterval}. In general, for a sequence $\varphi_1 ~ \langle op \rangle ~ \varphi_2 $ containing sub-sequences $\varphi_1$ and $\varphi_2$, where $\langle op \rangle ~\in$ \{{\tt\#\#[a:b], [*a], [*a:b]}\}, treating $\langle op \rangle$ as left-associative, when $\varphi_1$ is found to be true, CHAMS monitors simulation progress for a time of at least $b$ (or $a$, for {\tt[*a]}) time units before deciding the match/fail of the sequence.

	\begin{definition}
		\textbf{Depth of a Sequence Expression}
		\label{def:depth}
		The depth of a sequence expression $\varphi$, recursively defined as follows:
		
		\begin{tabular}{lcll}
			$\mathcal{D}(\psi)$ &=& 0\\
			$\mathcal{D}(\varphi ~[^*a])$ &=& $\mathcal{D}(\varphi) + |a|$ \\
			$\mathcal{D}(\varphi_1 ~\#\#[a:b]\varphi_2)$ &=& $\mathcal{D}(\varphi_1) + \mathcal{D}(\varphi_2) + b$\\
			$\mathcal{D}(\varphi_1 ~[^*a:b]\varphi_2)$ &=& $\mathcal{D}(\varphi_1) + \mathcal{D}(\varphi_2) + b$\\
		\end{tabular}
		
		where $\psi$ is a Boolean expression.
		~\hfill$\QEDA$
	\end{definition}
	
	For an assertion $\varphi_1$ {\tt |->} $\varphi_2$, evaluated over trace $\tau$, to decide the truth of 
	the assertion at time point $t$, where $t \in \mathcal{M}_E(s_1, \mathbf{I}(\tau))$, CHAMS allows 
	simulation to progress upto $t+ \mathcal{D}(\varphi_2)$.

\section{Empirical Studies}
\label{sec:results}

AMS assertion checking is relevant in two contexts, one in which the analog components are 
transistor level netlists, and one in which the analog components are replaced by behavioural 
models for accelerating simulation at the full-chip level. In order to study the overhead of
our tool CHAMS, we have used two implementations of a Low Dropout Regulator (LDO) as test cases,
namely a light-weight behavioural model written in Verilog-AMS, and an industry standard 
transistor level netlist of the same LDO. Several properties were coded in AMSAL, of which two are 
shown as illustrative examples.

\begin{property}
{\bf Settling time :} \textit{The settling time of the LDO should be less than 6~ms. The
settling time is defined as the time taken by the system to settle down within $ \pm 5\%  $ of 
the rated voltage $3.2V$ and stay there for at least 2~ms.}
		\begin{verbatim}
		property SettlingTime{};
		    @+{V(Vout),0.1*3.2} |-> ##[0:0.006] 
		    {V(Vout)>=0.95*3.2 && V(Vout)<=1.05*3.2}[*0.002];
		endproperty
		\end{verbatim}
\end{property}
	
\begin{property}
{\bf Power Sequencing:} \textit{The power domain sourced by the LDO should not be enabled until 
the output of the LDO remains within $ \pm 5\%  $ of the rated voltage $3.2V$ for at least 10~ms.
The enable signal for the power domain sourced by the LDO is called {\tt en}.}
\end{property}
		
\begin{lstlisting}[basicstyle=\ttfamily,columns=fullflexible,escapechar=`]
property Power_Sequencing{};
    ~{V(Vout)>3.15 && V(Vout)<3.25} |-> 
    ~{@+{en}}[*0.01];
endproperty
\end{lstlisting}

Table~\ref{tab:results_ldo} shows the results of our empirical studies. In addition to the
LDO, we also examined an industrial Buck Regulator netlist. All simulations were run using 
Cadence on a 2.33 GHz Intel-Xeon server with 32GB RAM. Column~3 represents the accuracy of the generated Verilog-AMS cross events used for monitoring the assertions. Both the Verilog-AMS BMOD and the transistor netlist of the circuit are run for the simulation time given in the table's second column. The standalone simulation time of the circuit, without assertion monitoring using CHAMS, is given in the fourth column. The fifth column shows the simulation time when the circuit is run with the assertion checker tool. A comparison of the last two columns indicates the overhead of the assertion monitoring.

\begin{table}[t]
		\centering
		\caption{Results from Empirical Studies}
 
		\label{tab:results_ldo}
		\setlength\tabcolsep{1.5pt}
		\renewcommand{\arraystretch}{0.9}
		\begin{tabular}{c}
			\toprule
			\\[-6pt]
			{\bf DESIGN STATISTICS }\\ [-6pt] \\
			\begin{tabular}{c c c c c c}
				\midrule
				 &
				{\bf \#Nodes} &
				{\bf \#Transistors} &
				{\bf \#Resistors} &
				{\bf \#Capacitors} &
				{\bf \#Diodes} \\
				{\bf LDO Netlist} & 1434 	& 336	& 1269 &	861 	& 6 \\
				{\bf Buck Netlist} & 1787 	& 2455	& 495 &	350 	& 67 \\
			\end{tabular}
			
			\\ \bottomrule
			\\ [-2pt]
			{\bf SIMULATION RESULTS } \\ [-6pt] \\
		   \toprule
		   \begin{tabular}{c c c c c c}
				& 
				&     
				& \multicolumn{3}{c}{\textbf{CPU Sim.Time}} \\ \cline{4-6}
				&
				&
				&
				&
				&
				\\
				\multirow{-3}{*}{\textbf{Assertions}}                                                                                 & \multirow{-3}{*}{\textbf{\begin{tabular}[c]{@{}c@{}}Simulation\\ Time\end{tabular}}} & \multirow{-3}{*}{\textbf{\begin{tabular}[c]{@{}c@{}}Cross Event\\ Accuracy \\ (sec, Volt)\end{tabular}}} & \multirow{-2}{*}{\textbf{Ckt. only}} & \multirow{-2}{*}{\textbf{\begin{tabular}[c]{@{}c@{}}Ckt. + \\ CHAMS\end{tabular}}} & \multirow{-2}{*}{\textbf{\begin{tabular}[c]{@{}c@{}}\%\\ Overhead\end{tabular}}} \\ \toprule
				\multicolumn{6}{c}{\textbf{LDO Behavioural Model in Verilog-AMS}}   \\
				\bottomrule
				
				&
				&
				&
				& 
				&                                                                                  \\
				&
				& \multirow{-2}{*}{1e-4, 1e-3}
				&
				& \multirow{-2}{*}{11m 05s} 
				& \multirow{-2}{*}{7.78} \\
				&
				&
				&
				&
				&
				\\
				
				& 
				& \multirow{-2}{*}{1e-6, 1e-4}
				& 
				& \multirow{-2}{*}{11m 57s}                    
				& \multirow{-2}{*}{16.21} \\
				&
				&
				& 
				&
				& \\
				\multirow{-6}{*}{\begin{tabular}[c]{@{}c@{}}RisingSequence,\\ Settling Time,\\ Overshoot,\\ Power\_Seq\end{tabular}} & 
				\multirow{-6}{*}{300ms} & 
				\multirow{-2}{*}{1e-9, 1e-6} & 
				\multirow{-6}{*}{10m 17s} & 
				\multirow{-2}{*}{13m 38s} &
				\multirow{-2}{*}{32.57} \\ 
			\toprule
				\multicolumn{6}{c}{\textbf{LDO Transistor Level Netlist}}
				\\
				\bottomrule
				&
				&
				&
				&
				& \\
				& 
				& \multirow{-2}{*}{1e-4, 1e-3}
				&
				& \multirow{-2}{*}{22m 46s}
				& \multirow{-2}{*}{10.61} \\
				&
				&
				&
				&
				& \\
				&
				& \multirow{-2}{*}{1e-6, 1e-4}
				& 
				& \multirow{-2}{*}{23m 42s}
				& \multirow{-2}{*}{15.14} \\
				&
				& 
				&
				&
				& 
				\\
				\multirow{-6}{*}{\begin{tabular}[c]{@{}c@{}}RisingSequence,\\ Settling Time,\\ Overshoot\end{tabular}}
				& \multirow{-6}{*}{40ms}
				& \multirow{-2}{*}{1e-9, 1e-6}
				& \multirow{-6}{*}{20m 35s}
				& \multirow{-2}{*}{26m 7s}
				& \multirow{-2}{*}{26.88} \\ 
			\toprule
			\multicolumn{6}{c}{\textbf{Buck Regulator Transistor Level Netlist}}   \\
			\bottomrule
			& & & & & \\
			&
			& \multirow{-2}{*}{1e-4, 1e-3}
				&
				& \multirow{-2}{*}{3h 10m 14s} 
				& \multirow{-2}{*}{1.16} \\
				&
				&
				&
				&
				&
				\\
				
				& 
				& \multirow{-2}{*}{1e-6, 1e-4}
				& 
				& \multirow{-2}{*}{3h 11m 07s}                    
				& \multirow{-2}{*}{1.63} \\
				&
				&
				& 
				&
				& \\
				\multirow{-6}{*}{\begin{tabular}[c]{@{}c@{}}RisingSequence,\\ Settling Time,\\ Overshoot\end{tabular}} & 
				\multirow{-6}{*}{250 $\mu$s} & 
				\multirow{-2}{*}{1e-9, 1e-6} & 
				\multirow{-6}{*}{3h 8m 3s} & 
				\multirow{-2}{*}{3h 17m 53s} &
				\multirow{-2}{*}{5.23} \\ \bottomrule
			\end{tabular}
		\end{tabular}
		
	\end{table}

	The overhead shown in the results are largely due to the multiple VPI callbacks that have to be executed 
	to accurately maintain the dense truth intervals for the PORVs. The major takeaway is that the assertion 
	monitoring is online, hence a failure will be reported at the very instance when the assertion fails. 
	Thus the designer can stop the simulating in-between whenever an assertion fails. Another appealing 
	feature of CHAMS is that it is a completely automated tool which requires minimum user intervention.

\section{Related Work}\label{sec:related}
In our past work on AMS-LTL we extended temporal logic to express assertions for AMS~\cite{AMS_LTL} and further proposed adding property variables to the logic in AMS-LTL$^L$~\cite{AMS_LTL-L}. We also describe mechanisms for AMS property checking using standard simulators~\cite{assertion-aware-sampling} and proposed using auxillary state machines to aid in testing and verification flows. In Ref.~\cite{debugging-window} we discuss how dense-time debugging windows for property matches and failures may be chosen and refined. We also introduced a language a mechanism for the analysis of features~\cite{dyfet,forfet}. Features are quantitative properties for AMS systems that specify a set of behaviours over which analog measurements are computed. The language of features is developed to be similar to SVA, which is a widely used property specification language for discrete domains in the Semiconductor industry. SVA, itself, was primarily developed for digital systems and therefore is limited in its capacity to express AMS properties.

In work proposed by other groups, the proposed languages either do not support recurrences or 
lack the ability to monitor properties online with off-the-shelf AMS circuit simulators. In 
Ref.~\cite{AMT, amt_tool}, an assertion monitoring tool is presented which analyzes output signals 
to compute the truth of assertions. The assertion language is developed by incorporating Timed 
Regular Expression (TRE)~\cite{TRE} into STL and suitable algorithms are presented for their match computation \cite{TRE_matching, TRE_derivative_based}. Although TREs offer a rich set of operators, they lack support for \textit{events}, \textit{delays}, and restricts $I$ in $\langle \varphi \rangle _I$ to have integer endpoints only. Authors in \cite{automata_skipping, automata_skipping_nfm} have designed an FJS-type online algorithm for computing whether a timed word matches a given specification, taken in the form of a timed automaton. In our approach, we use interval arithmetic and support a language with artifacts from circuits. Unlike other offerings, CHAMS  can interact with industrial circuit simulators to guarantee accuracy of property matches. Furthermore, given the fact that we primarily target the semiconductor industry as our use-case, and the fact that SVA is already extensively used in the field, adoption of an SVA-like language would be easier for verification engineers.
Other attempts have been made to express properties for AMS in PSL and evaluate these using 
modern analog simulators such as Spectre~\cite{prabal}. Assertions written on analog ports can 
be evaluated on a predefined digital clock, or at an analog event (such as a \textit{cross} event 
in Verilog-AMS) introducing analog to the assertion checking process. However, dense-time temporal 
properties can not be expressed in the proposed format.

Property languages for AMS over dense-time, developed using SVA-like syntax do not presently 
support specification of recurrence. Deciding property matches and failures for properties 
involving recurrent behaviours, that hold continuously over a period of time, requires a different 
mechanism when compared with those without recurrence.
	
\section{Conclusion}
\label{sec:conclude}
	
We believe that recurrence over continuous time is necessary to express many properties of AMS
designs. Our proposal for recurrence operators in AMS addresses this requirement. We deliberately
choose a syntax similar to SVA for ease of verification engineers conversant with SVA. However
properties expressed in our language AMSAL are evaluated by our tool CHAMS using interval 
arithmetic as opposed to cycle based evaluation of SVA. The ability of CHAMS to work with
standard commercial circuit simulators has facilitated its integration into the verification tool
flow of multiple semiconductor companies.

\balance


\begin{thebibliography}{}
    
    \bibitem{sva} ``SystemVerilog 3.1a Language Reference Manual"
    
    \bibitem{psl} ``IEEE Standard for Property Specification Language (PSL)"
	
	\bibitem{specification_patterns} S. Konrad and B. H. Cheng, ``Real-time specification patterns", \textit{Proceedings of the 27th international conference on Software engineering. ACM, 2005, pp. 372–381.}
	
	\bibitem{MTL} R. Alur and T. A. Henzinger, ``Real-time logics: Complexity and expressiveness", \textit{Information and Computation},vol. 104, no. 1, pp. 35–77, 1993.
	
	\bibitem{MITL} R. Alur, T. Feder, and T. A. Henzinger ``The Benefits of Relaxing Punctuality", \textit{J. ACM}, vol. 43, no. 1, pp. 116–146, Jan. 1996. [Online]. Available: http://doi.acm.org/10.1145/227595.227602
	
	\bibitem{TPTL} R. Alur and T. A. Henzinger, ``A Really Temporal Logic", \textit{J. ACM}, M, vol. 41, no. 1, pp. 181–203, Jan. 1994. [Online]. Available: http://doi.acm.org/10.1145/174644.174651
	
	\bibitem{localVariablesAntara} A. Ain and P. Dasgupta ``Interpreting Local Variables in AMS Assertions during Simulation", \textit{IEEE Transactions on Computer-Aided Design of Integrated Circuits and Systems, 2018}
    
    \bibitem{Mukhopadhyay:2009} R. Mukhopadhyay, S. K. Panda, P. Dasgupta, and J. Gough, ``Instrumenting AMS Assertion Verification on Commercial Platforms", \textit{ACM TODAES 2009}
    
    \bibitem{AMS_LTL} S. Mukherjee, P. Dasgupta, S. Mukhopadhyay, S. Little, J. Havlicek, and S. Chandrasekaran, ``Synchronizing AMS Assertions with AMS Simulation: From Theory to Practice", \textit{ACM Trans. Des. Autom. Electron. Syst}, vol. 14, no. 2, pp. 21:1–21:47, Apr. 2009. [Online]. Available: http://doi.acm.org/10.1145/1497561.1497564
    
    \bibitem{AMS_LTL-L} S. Mukherjee, P. Dasgupta, ``Incorporating Local Variables in Mixed-Signal Assertions", \textit{TENCON 2009} - 2009 IEEE Region 10 Conference, Singapore, 2009, pp. 1-5.
    
    \bibitem{assertion-aware-sampling} S. Mukherjee, P. Dasgupta, ``Assertion Aware Sampling Refinement: A Mixed-Signal Perspective", \textit{IEEE Transactions on Computer-Aided Design of Integrated Circuits and Systems}, vol. 31, no. 11, pp. 1772-1776, Nov. 2012.
    
    \bibitem{debugging-window} S. Mukherjee, P. Dasgupta, ``Computing Minimal Debugging Windows in Failure Traces of AMS Assertions", \textit{IEEE Transactions on Computer-Aided Design of Integrated Circuits and Systems}, 31.11 (2012): 1776-1781.
    
    \bibitem{dyfet} A.Ain, A.A.B.D. Costa, and P. Dasgupta, ``Feature Indented Assertions for Analog and Mixed-Signal Validation", \textit{IEEE Transactions on Computer-Aided Design of Integrated Circuits and Systems} 35.11 (2016): 1928-1941.
    
    \bibitem{forfet} A.A.B.D Costa, G. Frehse, and P. Dasgupta, ``Formal Feature Interpretation of Hybrid Systems", \textit{IEEE Transactions on Computer-Aided Design of Integrated Circuits and Systems} 37.11 (2018): 2474-2484.
    
    \bibitem{AMT} D. Ni{\v{c}}kovi{\'{c}}, O. Lebeltel, O. Maler, T. Ferr{\`e}re, D. Ulus, ``AMT 2.0: Qualitative and quantitative trace analysis with extended signal temporal logic", \textit{Tools and Algorithms for the Construction and Analysis of Systems 2018}
    
    \bibitem{amt_tool} D. Ni{\v{c}}kovi{\'{c}}, O. Maler,  ``AMT: A property-based monitoring tool for analog systems", \textit{International Conference on Formal Modeling and Analysis of Timed Systems}. Springer, Berlin, Heidelberg, 2007.
    
    \bibitem{temporal properties} Maler O., Nickovic D. (2004) Monitoring Temporal Properties of Continuous Signals. In: Lakhnech Y., Yovine S. (eds) Formal Techniques, Modelling and Analysis of Timed and Fault-Tolerant Systems. FTRTFT 2004, FORMATS 2004. Lecture Notes in Computer Science, vol 3253. Springer, Berlin, Heidelberg. https://doi.org/10.1007/978-3-540-30206-3\_12
    
    \bibitem{prabal}  P. Bhattacharya, D. O{’}Riordan, W. Hartong, ``Mixed signal assertion-based verification", \textit{ Design and Verification Conference and Exhibition, 2011}
    
    \bibitem{TRE} ``Timed regular expressions", Journal of ACM 2002 vol. 49, no. 2, pp. 172–206, Mar.2002. [Online]. Available: http://doi.acm.org/10.1145/506147.506151
    
    \bibitem{TRE_matching} Ulus D., Ferrère T., Asarin E., Maler O. (2014) Timed Pattern Matching. In: Legay A., Bozga M. (eds) Formal Modeling and Analysis of Timed Systems. FORMATS 2014. Lecture Notes in Computer Science, vol 8711. Springer, Cham. https://doi.org/10.1007/978-3-319-10512-3\_16
    
    \bibitem{TRE_derivative_based} Ulus D., Ferrère T., Asarin E., Maler O. (2016) Online Timed Pattern Matching Using Derivatives. In: Chechik M., Raskin JF. (eds) Tools and Algorithms for the Construction and Analysis of Systems. TACAS 2016. Lecture Notes in Computer Science, vol 9636. Springer, Berlin, Heidelberg. https://doi.org/10.1007/978-3-662-49674-9\_47
    
    \bibitem{automata_skipping} Waga, Masaki, Ichiro Hasuo, and Kohei Suenaga. “Efficient Online Timed Pattern Matching by Automata-Based Skipping.” Formal Modeling and Analysis of Timed Systems (2017): 224–243
    
    \bibitem{automata_skipping_nfm} Waga M., André É. (2019) Online Parametric Timed Pattern Matching with Automata-Based Skipping. In: Badger J., Rozier K. (eds) NASA Formal Methods. NFM 2019. Lecture Notes in Computer Science, vol 11460. Springer, Cham. https://doi.org/10.1007/978-3-030-20652-9\_26
    
    \end{thebibliography}
\end{document}